\documentclass[conference,10pt]{IEEEtran}

\makeatletter\def\input@path{{current}}\makeatother

\usepackage{amsmath,amssymb,amsfonts,amsthm}
\usepackage{algorithmic}
\usepackage{graphicx}
\usepackage{textcomp}
\usepackage{xcolor}
\usepackage[utf8]{inputenc}
\usepackage[margin=1in]{geometry}
\usepackage[T1]{fontenc}

\usepackage[english]{babel}
\usepackage[square, numbers]{natbib}                                                   
\usepackage{xfrac}

\def\BibTeX{{\rm B\kern-.05em{\sc i\kern-.025em b}\kern-.08em
    T\kern-.1667em\lower.7ex\hbox{E}\kern-.125emX}}

\newcommand{\blue}[1]{#1} 

\newcommand{\expminus}{e^{-\int_0^t \Qware(p) \, dp }}
\newcommand{\expplus}{e^{\int_0^\tau \Qware(p) \, dp }}
\newcommand{\bintegral}{\Fnominal \int_0^t \expplus \bonware(\tau) \, d \tau}  

\newcommand{\bonwareSuper}{b}
\newcommand{\malwareSuper}{m}

\newcommand{\functionality}{{{F}}}

\newcommand{\activity}{{{A}}}
\newcommand{\effectiveness}{{{E}}}

\newcommand{\unitstep}{u}

\newcommand{\malwareActivity}{{\activity}^{\malwareSuper}(t)}
\newcommand{\bonwareActivity}{{\activity}^{\bonwareSuper}(t)}
\newcommand{\malwareEffectiveness}{{\effectiveness}^{\malwareSuper}(t)}
\newcommand{\bonwareEffectiveness}{{\effectiveness}^{\bonwareSuper}(t)}
\newcommand{\malwareActivityConstant}{{\activity}^{\malwareSuper}}
\newcommand{\bonwareActivityConstant}{{\activity}^{\bonwareSuper}}
\newcommand{\malwareEffectivenessConstant}{{\effectiveness}^{\malwareSuper}}
\newcommand{\bonwareEffectivenessConstant}{{\effectiveness}^{\bonwareSuper}}

\newcommand{\malwareRate}{\theta^{\malwareSuper}}
\newcommand{\bonwareRate}{\theta^{\bonwareSuper}}
\newcommand{\malwareSize}{{{\gamma}^{\malwareSuper}}}
\newcommand{\bonwareSize}{{{\gamma}^{\bonwareSuper}}}

\newcommand\bonware{\mathcal{B}}
\newcommand\malware{\mathcal{M}}
\newcommand\Qware{\mathcal{Q}}
\newcommand\ma{\mathcal{A}} 

\newcommand\Fnominal{\blue{F_\text{N}}{}}

\newcommand\tchange{t^\star}   

\newcommand\bonQone{\frac{\bonware_1}{\Qware_1}}

\DeclareMathOperator{\erf}{erf}
\DeclareMathOperator{\unif}{Unif}
\DeclareMathOperator{\bern}{Bern}

\newcommand{\dunif}[2]{\unif\!\left(#1,#2\right)}
\newcommand{\dbern}[1]{\bern\!\left(#1\right)}

\newcommand\impact{\blue{impact}} 
\newcommand\Impact{\blue{Impact}} %

\newcommand\fff{}

\newcommand\aaa[1]{\color{black}\fff{\includegraphics[width=0.98\columnwidth]{#1}}
\vspace{-1mm}}

\newcommand\ppp[1]{\fff{\includegraphics[trim={0mm 0mm 0mm -3mm}, clip,width=0.98\columnwidth]{#1}}\vspace{-1mm}}
\newcommand\pppp[1]{\fff{\includegraphics[trim={0mm 0mm 0mm -3mm}, clip,width=0.98\columnwidth]{#1}}\vspace{1mm}}

\IEEEoverridecommandlockouts
\begin{document}

\graphicspath{ {./figures/} }

\newcommand\combined{ \hspace{-4.5mm} The views and conclusions
  contained in this document are those of the authors and should not
  be interpreted as representing the official policies, either
  expressed or implied, of the Army Research Laboratory or the
  U.S. Government. The U.S. Government is authorized to reproduce and
  distribute reprints for Government purposes notwithstanding any
  copyright notation herein.  This work was partially funded by Cyber
  Technologies, Deputy CTO for Critical Technologies/Applied
  Technology, Office of the Under Secretary of Defense Research and
  Engineering.  Dr. Vandekerckhove's research was sponsored by the
  Army Research Laboratory and was accomplished under Cooperative
  Agreement Number W911NF-21-2-0284.}

\title{Mathematical Modeling of Cyber Resilience\thanks{\combined}}

\author{\IEEEauthorblockN{Alexander Kott, Michael J. Weisman}
\IEEEauthorblockA{
  \textit{U.S. Army Combat Capabilities Development Command}\\
  \textit{Army Research Laboratory} \\
  Adelphi, MD \\
  \{alexander.kott1, michael.j.weisman2\}.civ@army.mil
}
\and
\IEEEauthorblockN{Joachim Vandekerckhove}
\IEEEauthorblockA{\textit{University of California, Irvine} \\
  \textit{Department of Cognitive Sciences}\\
  Irvine, CA \\
  joachim@uci.edu}
}

\maketitle

\begin{abstract}
\label{abstract}
We identify quantitative characteristics of responses to cyber compromises that can be learned from repeatable, systematic experiments. We model a vehicle equipped with an autonomous cyber-defense system and which also has some inherent physical resilience features.  When attacked by malware, this ensemble of cyber-physical features (i.e., “bonware”) strives to resist and recover from the performance degradation caused by the malware’s attack.  We propose parsimonious continuous models, and develop stochastic models to aid in quantifying systems' resilience to cyber attacks.
\end{abstract}

\section{Introduction}

Resilience continues to gain attention as a key property of cyber and cyber-physical systems, for the purposes of cyber defense. Although definitions vary, it is generally agreed that cyber resilience refers to the ability of a system to resist and recover from a cyber compromise that degrades the mission-relevant performance of the system \cite{kott2019cyber}.  Resilience should not be conflated with risk or security \cite{linkov2018risk}.

To make the discussion more concrete, consider the example of a military ground logistics vehicle, possibly unmanned, which performs a mission of delivering heavy supplies along a difficult route. The adversary's malware successfully gains access to the Controller Area Network (CAN bus) of the vehicle \cite{bozdal2018august}. Then, the malware executes cyber attacks by sending a combination of messages intended to degrade the vehicle's performance and diminish its ability to complete its delivery mission. We assume that the malware is at least partly successful, and the vehicle indeed begins to experience a degradation of its mission-relevant performance.

At this point, we expect the vehicle's resilience-relevant elements to  resist the degradation and then to recover its performance to a satisfactory level, within an acceptably short time period. These ``resilience-relevant elements'' might be of several kinds. First, because the vehicle is a cyber-physical system, certain physical characteristics of the vehicles mechanisms will provide a degree of resilience. For example, the cooling system of the vehicle will exhibit a significant resistance to overheating even if the malware succeeds in misrepresenting the temperature sensors data. Second, appropriate defensive software residing on the vehicle continually monitors and analyzes the information passing through the CAN bus \cite{kott2018}.  When the situation appears suspicious, it may take actions such as blocking or correcting potentially malicious messages. Third, it is possible that a remote monitoring center, staffed with experienced human cyber defenders, will detect a cyber compromise and will provide corrective actions remotely \cite{kott2021cyber}.

For the purposes of this paper, we assume that the remote monitoring and resilience via external intervention is impossible \cite{kott2020doers}. This may be the case if the vehicle must maintain radio silence for survivability purposes, or if the malware spoofs or blocks communication channels of the vehicle. Therefore, in this paper we assume that resilience is provided by the first two classes of resilience-relevant elements. Here, by analogy with malware, we call these ``bonware'' -- a combination of physical and cyber features of the vehicle that serve to resist and recover from a cyber compromise.

A key challenge in the field of cyber resilience is quantifying or measuring resilience. Indeed, no engineering discipline achieved significant maturity without being able to measure the properties of phenomena relevant to the discipline \cite{kott2021cyber}. Developers of systems like the notional vehicle in our example must be able to  quantify the resilience of the vehicle under development in order to know whether the features they introduce in the vehicle improve its cyber resilience, or make it worse. 
Similarly, buyers of the vehicle need to know how to quantitatively specify and test resilience in order to determine whether the product meets their specifications.  

In this paper, we report some of the results of a project called \textit{Quantitative Measurement of Cyber Resilience} (QMoCR) in which our research team seeks to identify quantitative characteristics of systems' responses to cyber compromises that can be derived from repeatable, systematic experiments. Briefly, we have constructed a test-bed in which a surrogate vehicle is subjected to controlled cyber attacks produced by malware. The vehicle is equipped with an autonomous cyber-defense system \cite{kott2020doers,kott2018} and also has some inherent physical resilience features. This ensemble of cyber-physical features (i.e., ``bonware'') strives to resist and recover from the performance degradation caused by the malware's attack. The test bed is instrumented in such a way that we can measure observable manifestations of this battle between the malware and bonware, especially the mission-relevant performance parameters of the vehicle. 

The details of the test bed and the experiment are given in a companion paper \cite{ellis2022experimental}. The focus of this paper is different -– here we concentrate on constructing mathematical models that can be used to describe the dynamics of the malware-bonware battle. We seek models that are parsimonious in the number of empirical parameters and allow us to easily derive parameters of the model from experimental data.

The remainder of the paper is organized as follows. In the next section, we briefly describe prior work related to quantification of cyber resilience.  We provide formal definitions of accomplishment and functionality.  We propose a class of parsimonious models in which effects of both malware and bonware are approximated as deterministic, continuous differentiable variables, and we explore several variations of such models. In the following section we propose a different class of models – stochastic models, and we show how this class is related to the previously proposed class of deterministic models. Then we show how these models are used to approximate experimental data obtained with our surrogate vehicle. We show how to determine the parameters of the models from experimental data.  We discuss whether these parameters might be considered quantitative characteristics (i.e., measurements) of the bonware's cyber resilience.    

\section{Prior work}\label{sec:prior-work}

A growing body of literature explores quantification of resilience in general and cyber resilience in particular. Very approximately, the literature can be divided into two categories: 
(1) qualitative assessments of a system (actually existing or its design) by subject matter experts (SMEs) \cite{alexeev}, \cite{henshel} and 
(2) quantitative measurements based on empirical or experimental observations of how a system (or its high-fidelity model; \cite{kott2017assessing}) responds to a cyber compromise \cite{kott2019cyber,ligo2021how}.
In the first category, a well-cited example is the approach called the cyber-resilience matrix \cite{linkov2013resilience}. In this approach, a system is considered as spanning four domains: (1) physical (i.e., the physical resources of the system, and the design, capabilities, features and characteristics of those resources); (2) informational (i.e., the system's availability, storage, and use of information); (3) cognitive (i.e., the ways in which informational and physical resources are used to comprehend the situation and make pertinent decisions); and (4) social (i.e., structure, relations, and communications of social nature within and around the system). For each of these domains of the system, SMEs are asked to assess, and to express in metrics, the extent to which the system exhibits the ability to (1) plan and prepare for an adverse cyber incident; (2) absorb the impact of the adverse cyber incident; (3) recover from the effects of the adverse cyber incident; and (4) adapt to the ramifications of the adverse cyber incident. In this way, the approach defines a 4-by-4 matrix that serves as a framework for structured assessments by SMEs.     

Another example within the same category (i.e., qualitative assessments of a system by SMEs) is a recent, elaborate approach proposed by \cite{beling2021developmental}. The approach is called Framework for Operational Resilience in Engineering and System Test (FOREST), and a key methodology within FOREST is called Testable Resilience Efficacy Elements (TREE). For a given system or subsystem, the methodology requires SMEs to assess, among others, how well the resilience solution is able to (1) sense or discover a successful cyber-attack; (2) identify the part of the system that has been successfully attacked; (3) reconfigure the system in order to mitigate and contain the consequences of the attack. Assessment may include tests of the system, although the methodology does not prescribe the tests. 

Undoubtedly, such methodologies can be very valuable in finding opportunities in improvements of cyber-resilience in a system that is either at the design stage or is already constructed. Still, these are essentially qualitative assessments, not quantitative measurements derived from an experiment. 

In the second category (i.e., quantitative measurements based on empirical or experimental observations of how a system, or its high-fidelity model, responds to a cyber compromise), most approaches tend to revolve around a common idea we call here the area under the curve (AUC) method \cite{hosseini2016review,kott2021to}.
In an experiment/test, a system is engaged into a performance of a representative mission, and then is subjected to an ensemble or sequence of representative cyber attacks. A mission-relevant quantitative functionality of the system is observed and recorded. The resulting average functionality, divided by normal functionality, can be used as a measure of resilience.  



However, AUC-based resilience measures are inherently cumulative, aggregate measures, and do not tell us much about the underlying processes. For example, is it possible to quantify the resilience \impact{} of the bonware of the given system? Similarly, is it possible to quantify the \impact{} of malware? In addition, is it possible to gain insights into how these values of \impact{fulness} vary over time during an incident? We offer steps toward answering such questions.  

\section{Quantitative measures of cyber resilience}\label{sec:qmocr}


Every mission has a goal, and we postulate that for a given mission, there exists a function $\ma(t)$ that represents accomplishment and is cumulative from the mission start time up until the present time $t$.  We define functionality, $\functionality(t)$, to be the time derivative of mission accomplishment.  Thus,
\begin{equation}
  \label{eq:ma}
  \functionality(t) = \frac{d \ma }{dt}, \quad \ma(t) = \int_{t_0}^t \functionality(\tau) \, d\tau.
\end{equation}
The normal functionality, when the system performs normally and does not experience effects of a cyber attack, may, in general, vary with time. For simplicity, throughout this paper, we assume the normal functionality to be constant in time, $\Fnominal(t)=\Fnominal$. Thus, the functionality of our system prior to an attack or other malfunction is normal at the start time: 
$\functionality(t_0)=\Fnominal$.

  In the following sections we develop models for the behavior of a system's functionality over the course of a mission during which it is being attacked by malware and defended by bonware.

In the first set of models, we assume that there is an observable, sufficiently smooth function representing mission accomplishment, and we define functionality to be its time derivative.  Then, we motivate a parsimonious model for the differential equation governing functionality, give the general solution, and discuss a few specific cases.  We  then develop a stochastic model and show its relationship to the continuous model.  Finally, we show, both analytically and empirically, that with sufficiently many instantiations, the average of an ensemble of stochastic curves will approximate the solution to the continuous model differential equation.

\section{Continuous model}\label{sec:continuous}

For the first set of models, we make the assumption that mission accomplishment is twice continuously differentiable: $\ma\in C^2$, and thus $\functionality \in C^1$. 

As a first approximation, we let the \impact{} of malware on the
derivative of functionality be linear. 
The \impact{} of bonware is similarly defined and proportional to the
level of functionality below normal.
Malware degrades the system while bonware aims to increase
functionality over time.  Malware \impact{} and bonware \impact{} are
assumed to be continuous functions of time,
$\malware, \bonware \in C^0.$ The \impact{} on functionality is the sum
of the \impact{s} of malware and bonware, and 
\begin{equation}
  \frac{d\functionality}{dt} + \Qware(t) \functionality(t) = \Fnominal \bonware (t),
\label{eq:00}
\end{equation}
where $\Qware(t)=\malware(t)+\bonware(t).$
Since we expect bonware to help (or at least not harm) and malware to
not help, we assume $\bonware(t) \ge 0$ and $\malware(t) \ge 0$.  We also assume normal
functionality is positive, $\functionality_0 > 0,$ and functionality
is always nonnegative and less than or equal to normal functionality,
$0 \le \functionality(t) \le \Fnominal.$ 
This first-order linear differential equation has the following solution: 
\small
\begin{equation}
  F(t)  = \expminus \left( F(0) + \bintegral \right).
\end{equation}
\normalsize
To help us understand how the model works, we find explicit solutions for a number of examples.  

  \subsection{Constant model}

Assuming $\malware, \bonware,$ and $\Qware$ are constant, we have
\begin{equation}  \label{eq:1}
   \frac{d\functionality}{dt} + \Qware \functionality(t) =  \Fnominal \bonware.
\end{equation}

\subsubsection{No bonware}

If $\bonware=0$, then Equation~\ref{eq:1} reduces to $\frac{d\functionality}{dt} + \malware \functionality(t) =  0$ and $\functionality(t) = \functionality(0) e^{-\malware t}$. If also $\malware=0$ (no bonware and no malware), then $\frac{d\functionality}{dt}=0$ and $\functionality(t)=\functionality(0)$.

\subsubsection{Bonware}\label{sec:4.2}

With bonware present, 
the solution is 
\begin{equation} \label{eq:3}
    \functionality(t) = \left[\functionality(0) - \frac{\Fnominal \bonware}{ \Qware } \right] e^{-\Qware t} + \frac{\Fnominal \bonware}{\Qware}.
\end{equation}

\begin{figure}[th]
  \centering
  \aaa{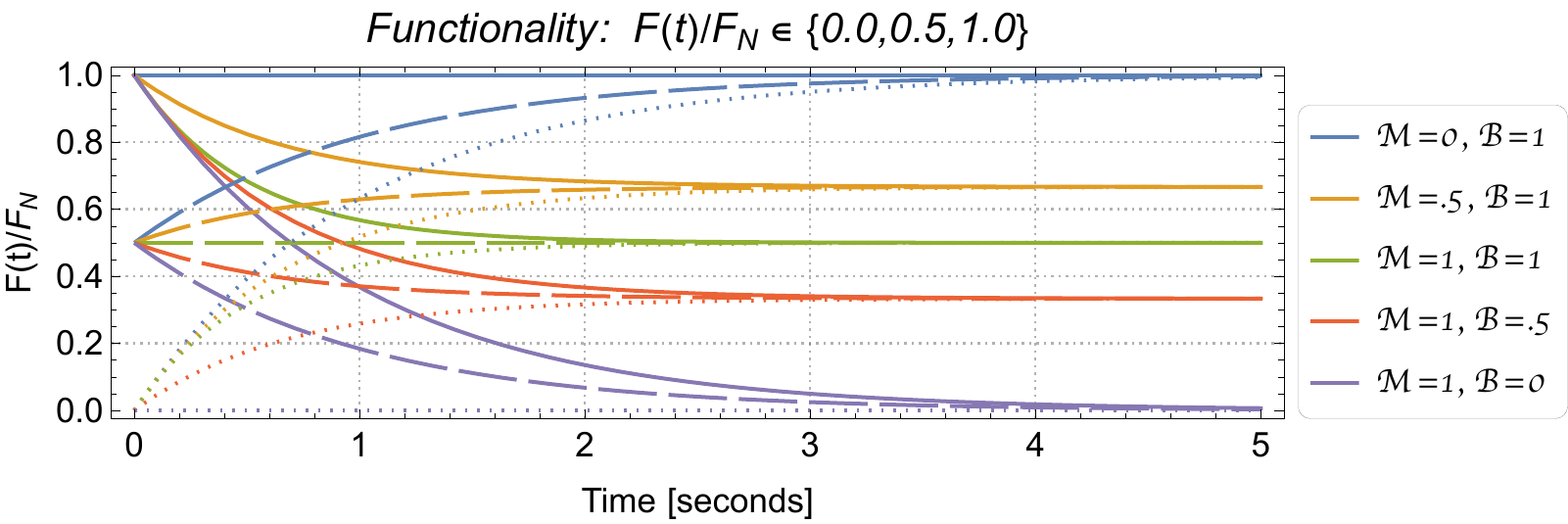}
    \caption{Normalized functionality, $\functionality(t)/\Fnominal$, is shown for various values of $\malware$ (malware attacking) and $\bonware$ (bonware defending) and at initial conditions $\functionality(0)/\Fnominal \in\{0.0, 0.5, 1.0 \}$.  The functionality over time depends on the relative strengths of bonware and malware and on the intial condition.  When the system initially is at normal functionality and when malware overpowers bonware, functionality exhibits exponential decay.  When functionality initially is low, and when bonware overpowers malware, the system recovers (via Eq.~\ref{eq:3}) to $\frac{\Fnominal \bonware}{\malware+\bonware}$.}
    \label{fig:2}
\end{figure}

If $\functionality(0)>\sfrac{\Fnominal\bonware}{\Qware},$ then $\functionality(t)$ will initially, at time $t=0$, be at $\functionality(0)$ and decrease to $\sfrac{\Fnominal\bonware}{\Qware}$.  If  $\functionality(0)>\sfrac{\Fnominal\bonware}{\Qware},$ then the function $\functionality(t) = \functionality(0)$ will be constant.  If $\functionality(0)<\sfrac{\Fnominal \bonware}{\Qware},$ the function will start at $\functionality=\functionality(0)$ and increase to $\sfrac{\Fnominal \bonware}{\Qware}.$ Examples of these situations are shown in Figure~\ref{fig:2}.  The plot of $(\malware>0)$ in Figure~\ref{fig:2} shows that even in the presence of bonware, malware will still have an impact on the system.  \blue{The steady-state of the system is obtained either by setting $\frac{d\functionality}{dt}=0$ in Equation~\ref{eq:1} or letting $t\to\infty:$}
\begin{equation} \label{eq:5}
  \functionality_\infty=  \lim_{t\to\infty} \functionality(t)  =\Fnominal\frac{\bonware}{\malware+\bonware}
\end{equation}
so that the antidote to malware is to overwhelm it with bonware.  \blue{The exponent, $-\Qware t = (-\malware-\bonware)t$ in the solution given by Equation $\ref{eq:3}$ indicates that increasing the \impact{} of either malware or bonware will cause the system to more quickly approach steady-state.}{}  At steady-state,
\begin{equation}\label{eq:5a}
  \begin{aligned}
    \frac{ \Fnominal -\functionality_\infty}{\functionality_\infty}  =  \frac{\malware}{\malware+\bonware}.
  \end{aligned}
\end{equation}
Equation~\ref{eq:5a} gives us further insight into the trade-off between impacts of both malware and bonware.  The relative decrease of the function from normal functionality is equal to the ratio of malware \impact{} to the sum of malware and bonware \impact{}s.
  \subsection{Piecewise constant model}

If either malware's or bonware's impact diminishes at some point in the incident, the model may switch from one set of constants defining malware and bonware to another set of constants.  The differential equation (Eq.~\ref{eq:00}) may now be expressed as
\begin{equation}
  \frac{d\functionality}{dt} = \sum_{j=0}^{N-1}(\Fnominal-\functionality(t)) \bonware_j(t) -  \functionality(t) \malware_j (t),
  \label{eq:000}
\end{equation}
where the vectors $\boldsymbol{\malware} =( {\malware}_0,
{\malware}_1, \cdots {\malware}_{N-1} )$ and $\boldsymbol{\bonware} =(
{\bonware}_0,  {\bonware}_1,\cdots, {\bonware}_{N-1})$ contain the
malware \impact{} and bonware \impact{}s within time windows whose end
points are defined by  $\{t_0, t_1, \cdots, t_N \}$. The solution will
be a function which, in each time interval, is the solution found in
Equation~\ref{eq:3}.  The purple curve in Figure \ref{fig:notional} is
a realization of this model.

\begin{figure}[t] 
  \centering
  \aaa{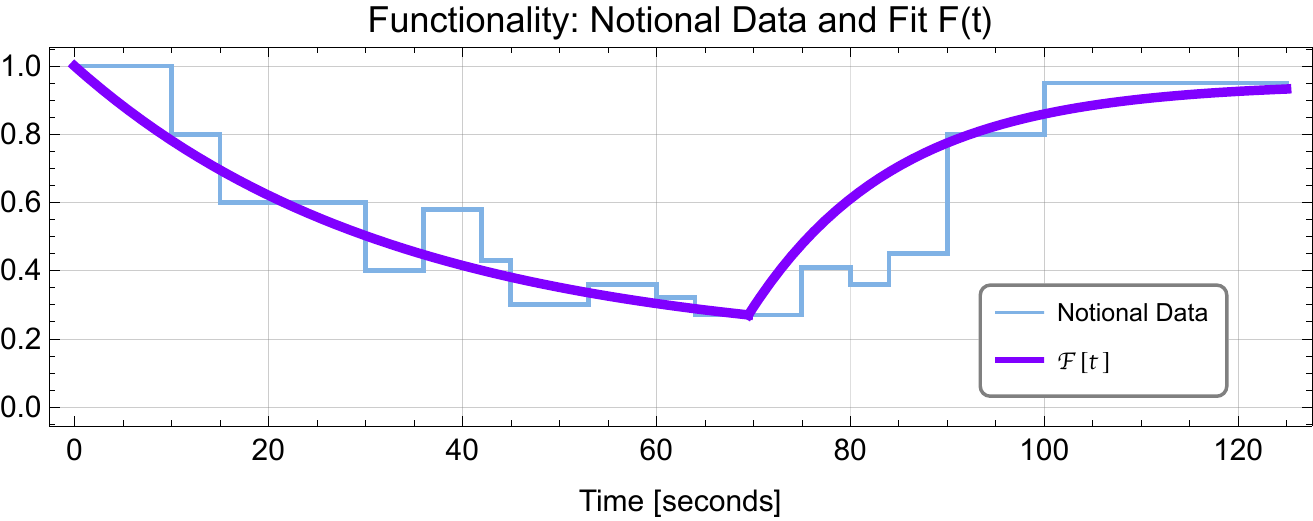}
  \caption{The smooth line is an example functionality curve with piecewise constant malware and bonware impacts.  The notional data and piecewise constant model fit are described below in Section \ref{sec:parameters}.}
  \label{fig:notional}
\end{figure}

  \subsection{Linear model}  \label{sec:linear}

\newcommand{\tempa}{\Omega(t)}
\newcommand{\tempb}{\Lambda}

The \impact{}s of malware and bonware may also be linear functions of $t$, so that $\malware(t) = \nu - \mu t$,    $\bonware(t) = \alpha-\beta t$, and  $\Qware(t)   =  \lambda - \omega t$, where $\lambda = \alpha + \nu$ and $\omega  = \beta + \mu$.  Under this linear model, Equation~\ref{eq:00} becomes
\begin{equation}
  \frac{d\functionality}{dt} + (\lambda - \omega t) \functionality(t) = \Fnominal (\alpha - \beta t)
\label{eq:linear:model}
\end{equation}
The solution can be expressed in terms of the error function
$\blue{\erf(z)=\frac{2}{\sqrt{\pi }}\int_0^z e^{-\tau^2}\,d\tau}$:
\begin{equation}
  \label{eq:6}
  \begin{aligned}
    \frac{\functionality(t)}{\Fnominal} &= \frac{1}{\tempa} \left\{
      \frac{\functionality(0)}{\Fnominal}
      -\frac{\beta}{\omega}\left(1-{\tempa}\right) +(\alpha \omega
      -\beta \lambda )
    \right.\\
    \times&\left.  \frac{\sqrt{\frac{\pi }{2}} e^{\tempb^2} }{\omega
        ^{3/2}}\left[\erf\left(\tempb\right)+\erf\left(\frac{\omega
            t}{\sqrt{2 \omega }}-\tempb\right)\right] \right\}
\end{aligned}
\end{equation}
\normalsize
where $\tempa = e^{\lambda  t-\frac{1}{2}\omega t^2}$, and $\tempb = \sfrac{\lambda}{\sqrt{2 \omega }}$. 

  \subsection{Piecewise linear model}

\newcommand{\tcj}{\Omega_j(t)}
\newcommand{\td}{\Lambda}

Both malware and bonware \impact{}s may initially be linear, but if
the situation changes and a different linear model holds after a time,
the model should be able to account for it.  In particular, if malware
\impact{} is decreasing over time, at some point we will reach
$\malware=0$ and the model switches to a new linear model.
Equation~\ref{eq:linear:model} can be written
\begin{equation*}
  \frac{d\functionality}{dt} = \sum_{j=0}^{N-1}  \left[ (\lambda_j - \omega_j t) \functionality(t) -  \Fnominal (\alpha_j - \beta_j t) \right]. \label{eq:piecewise:linear}
\end{equation*}
The solution follows from Equation~\ref{eq:6}. 
Example realizations of the piecewise linear models are shown in
Figure~\ref{fig:piecewise:linear}.  The shapes of the curves resemble experimental data discussed in \cite{ellis2022experimental}.

\begin{figure}[t] 
  \centering
    \aaa{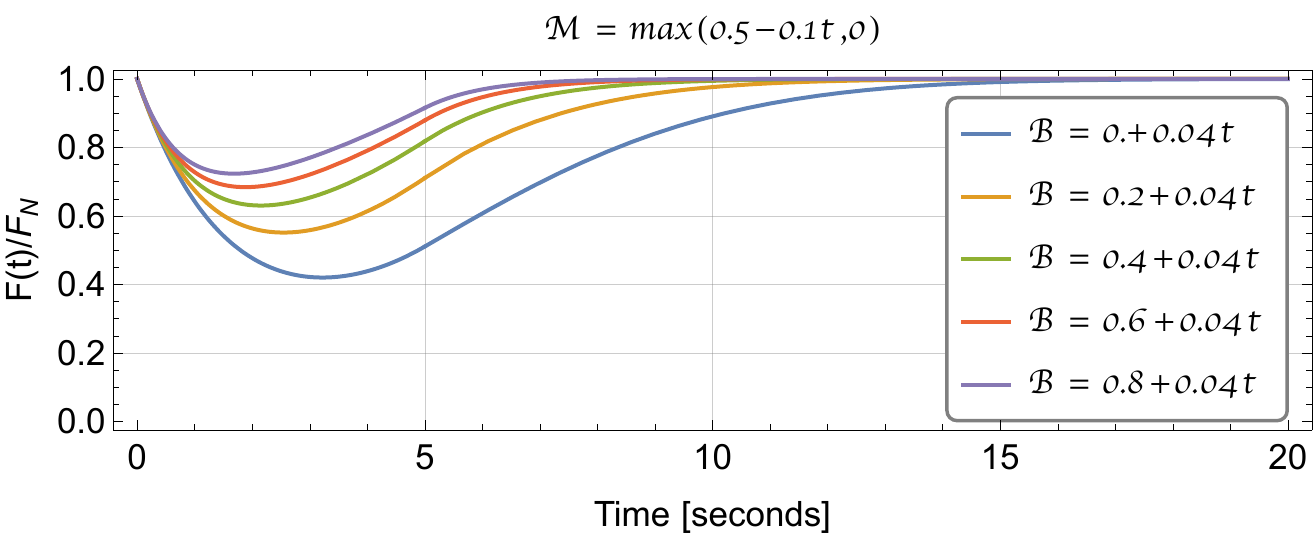}
    \caption{Normalized functionality, $\functionality(t)/\Fnominal$,
      for piecewise linear models. \Impact{}s of malware and bonware
      are linear functions of time.}
    \label{fig:piecewise:linear}
\end{figure}

\section{Stochastic differential equation model}

In this section, we extend the previously defined differential equation (DE) model to a \textit{stochastic} differential equation (SDE) model.  The extension is motivated by the discontinuous nature of the notional data in Figure~\ref{fig:notional}.  Whereas the DE model assumed a smooth functionality curve, our stochastic version allows for a more punctuated attack-and-restoration pattern.

In the SDE model, both malware and bonware may be active at random (or a priori unknown) times with random (or a priori unknown) effectiveness. Let malware activity $\malwareActivity \in \{0,1\}$ indicate whether malware was successful at time $t$, bonware activity $\bonwareActivity \in \{0,1\}$ indicate whether bonware was successful at time $t$, malware effectiveness $\malwareEffectiveness \in [0,1)$ express the proportion of functionality reduced by the malware's success at time $t$, and bonware effectiveness $\bonwareEffectiveness \in [0,1)$ express the proportion of damage undone by the bonware's success at time $t$. The SDE analog to Equation~\ref{eq:00} is then
\begin{equation}\label{eq:sde:generic}
  \frac{d \functionality}{d t} 
  =  \left(\Fnominal-\functionality(t)\right) \bonwareActivity \bonwareEffectiveness - \functionality(t) \malwareActivity \malwareEffectiveness.   
\end{equation}

Rather than changing deterministically over time, these model parameters are assumed to vary stochastically according to these distributions: 
\begin{eqnarray}
  \malwareActivity      &\sim& \dbern{\malwareRate(t)}\label{eq:sde:mwa}, \\
  \bonwareActivity      &\sim& \dbern{\bonwareRate(t)}\label{eq:sde:bwa}, \\
  \malwareEffectiveness &\sim& \dunif{0}{\malwareSize(t)}\label{eq:sde:mwe},\\
  \bonwareEffectiveness &\sim& \dunif{0}{\bonwareSize(t)}\label{eq:sde:bwe},
\end{eqnarray}
where $\dbern{\theta}$ indicates the Bernoulli distribution with rate $\theta$ and $\dunif{0}{\gamma}$ indicates a uniform distribution with lower bound $0$ and upper bound $\gamma$.  Hence, $\malwareRate(t) \in {[0,1]}$ is the probability that malware is successful at time $t$, $\bonwareRate(t) \in {[0,1]}$ is the probability that bonware is successful at time $t$, $\malwareSize(t) \in {(0,1]}$ is the maximum fraction of damage inflicted by malware, and $\bonwareSize(t) \in {(0,1]}$ is the maximum fraction of damage undone by bonware.

Like the ordinary differential equation (ODE) model, the SDE model allows for a number of interesting variants.  In the remainder of this section, we introduce some useful simplifications and extensions.

  \subsection{Constant parameters}

In the simplest version of the SDE model, we assume that the rate of malware attacks and their maximum efficiency are constant in time, and that the rate of bonware restoration and its maximum efficiency are both constant in time. Specifically, under this model it is assumed that $\malwareRate(t) = \theta^{\malwareSuper}$, $\bonwareRate(t) = \theta^{\bonwareSuper}$, $\malwareSize(t) = \gamma^{\malwareSuper}$, and $\bonwareSize(t) = \gamma^{\bonwareSuper}$ (cf.\ Eqs.~\ref{eq:sde:mwa}--\ref{eq:sde:bwe}).

This version of the model is parsimonious, with only four free parameters, each with a useful interpretation.



  \subsection{Piecewise constant parameters}

In a first extension of the SDE model, we assume that malware and bonware have an activity rate of 0 until they activate at time points $t^{\malwareSuper}$ and $t^{\bonwareSuper}$, respectively. After activation, both activity rates are constant.  Additionally, both efficiency parameters are assumed to be constant in time. Specifically,  $\malwareRate(t) = \theta^{\malwareSuper} \unitstep\!\left(t-t^{\malwareSuper}\right)$, $\bonwareRate(t) = \theta^{\bonwareSuper} \unitstep\!\left(t-t^{\bonwareSuper}\right)$, $\malwareSize(t) = \gamma^{\malwareSuper}$, and $\bonwareSize(t) = \gamma^{\bonwareSuper}$.
Here $\unitstep(\cdot)$ is the unit step function.

The piecewise constant formulation here is particularly apposite for our experiments, in which the timing of malware attacks is known.  The analyst can then choose between considering $t^{\malwareSuper}$ and $t^{\bonwareSuper}$ as known, and retain the parsimony of a four-parameter model, or considering them unknown and estimate them from data in a six-parameter model.  The difference between the known onset time of an attack and the estimated $t^m$ may then be interpreted as the time it takes for an attack to take effect. Similarly, the difference between that onset time and the estimated $t^b$ may be interpreted as the delay until bonware begins to restore functionality after an attack.

\subsection{Parameter expansion of the SDE model}

The SDE model can be conveniently stated as a hidden Markov model and implemented as a directed acyclic graph \cite{griffiths2008bayesian} for efficient parameter estimation with a general-purpose Bayesian inference engine (e.g., JAGS; \cite{plummer2003jags}).  Our implementation relied on a sequential definition for the likelihood function:
$\left(\functionality(t+1) \mid \functionality(t), \ldots\right) \sim  \dunif{L(t)}{U(t)}$, with $L(t) = \functionality(t)-\malwareActivity\malwareEffectiveness\functionality(t)$ and $U(t) = \functionality(t)+\bonwareActivity\bonwareEffectiveness\left(1-\functionality(t)\right)$.

Since this likelihood function depends on the unknown stochastic parameters $\malwareActivity$, $\bonwareActivity$, $\malwareEffectiveness$, and $\bonwareEffectiveness$, we applied a parameter expansion approach \cite{liu1998parameter,gelman2004parameterization} using Equations~\ref{eq:sde:mwa}--\ref{eq:sde:bwe}. 

\subsection{Relationship between continuous and SDE model}\label{sec:proposition}

With the parameters of the stochastic model selected appropriately, we show that as the number of stochastic realizations increases, the expectation of the solution to the stochastic differential equation model approaches that of the ODE model.   We show this for the simple constant parameter case.  The general result follows by extension.

\newtheorem*{theorem}{Theorem}

\begin{theorem}
Let $y^m_k \sim \dbern{2\malware}$, $y^b_k \sim
\dbern{2\bonware}$, $z^m_{k} \sim \dunif{0}{\functionality_{k}}$, $z^b_{k} \sim \dunif{0}{\Fnominal-\functionality_{k}}$, and 
\begin{equation}  \label{eq:stochastic:diff}
  \functionality_{k+1}  =  \functionality_{k} 
     - y^m_{k} z^m_{k}
     + y^b_{k} z^b_{k}, \quad (k=1,\hdots,K). 
\end{equation}
Let ${\mathcal{F}_k}_n=\frac{{\functionality_k}_j}{n}$, $(j=1,\hdots,n)$, then\\
\blue{$\mathcal{F}_k~=~\mathbb{E} (\functionality_k)~=~\lim_{n\to\infty}
 {\mathcal{F}_k}_n$ and $\mathcal{F}_k\approx \functionality(k),$ for
 large $k$, 
 where $\functionality(t)$ is the solution to the initial value problem
given by Equation~\ref{eq:1} with $F(0)=\mathcal{F}_0.$}
\end{theorem}

\begin{proof}
\blue{Take the expectation of Equation~\ref{eq:stochastic:diff}.  Then
  $\mathcal{F}_k~-~\mathcal{F}_{k-1}+ (\malware+\bonware)
  \mathcal{F}_{k-1}=\Fnominal \bonware$.  With
  ${\mathcal{F}_0}_n=\mathcal{F}_0,$
 the~solution~is 
$\mathcal{F}_k = \left[\mathcal{F}_0 - \frac{\Fnominal
    \bonware}{ \Qware } \right] (1-\Qware)^k + \frac{\Fnominal
  \bonware}{\Qware},$ which approximates Equation~\ref{eq:3} for large
$k$.}
\end{proof}

\section{An application of the model}
\subsection{Obtaining model parameters}
\label{sec:parameters}
%
%
Given notional data that represents a typical curve of functionality
over the course of an incident where malware and bonware are active,
we develop a fast method to estimate the continuous model parameters
for a curve that approximates the data, and use these parameters to
generate further realizations based on this model.  In Figure
\ref{fig:notional} an example of such notional data is plotted (in
light blue).  In this section, we illustrate our method to
extract the model parameters from this curve.



The set $P=\{t_0, \hdots, t_K\}$ partitions the mission timeline and
malware and bonware are constant in each interval $(t_{i-1},t_i), \,\,
i=1, \hdots, K. \,\,$  In each interval,
$\Qware_i=\malware_i+\bonware_i$ and the differential equation
governing Continuous Model I is $\frac{dF(t)}{dt}+\Qware_i F(t) =
\Fnominal (t) \bonware_i.$  Thus, in each interval $(t_{i-1},t_i),$
the solution is 
\begin{equation*}
  \functionality(t)     =
  \left[\functionality(t_{i-1}) - \frac{\Fnominal \bonware_i}{
      \Qware_i } \right]  e^{-\Qware_i (t-t_{i-1})} + \frac{\Fnominal \bonware_i}{\Qware_i}.
\end{equation*}
We compute the effectiveness, $\malwareEffectivenessConstant$, and activity, $\malwareActivityConstant$, of malware
and of bonware ($\bonwareEffectivenessConstant, \bonwareActivityConstant$), in each interval:
$\malware_i = \malwareEffectivenessConstant_i \malwareActivityConstant_i$, $\bonware_i = \bonwareEffectivenessConstant_i \bonwareActivityConstant_i$.

We observe that there is a unique switching time $\tchange$ where the
functionality's trend reverses, and thus we take $K=2.$ Before the
switch, the impact of malware is greater than that of bonware. From
the time of the switch until the end of the mission, bonware is
stronger.  To estimate the switching time $\tchange$, we find the
minimum of the data to occur over the interval from 64 s to 75 s.
There, the minimum value of the data curve is $m=0.27$.  Taking the
midpoint, our estimate for $\tchange$ is 69.5 s.  We estimate the
activity of malware before switching to be the number of times the
data curve decreases divided by the switching time.  Similarly, our
estimate of bonware activity is the number of times the data curve
increases prior to the switching time.  We thus have
$\malwareActivityConstant_1~\approx~\sfrac{7}{69.5}~\approx 0.101$ and
$\bonwareActivityConstant_1 \approx \sfrac{2}{69.5} \approx 0.014$.

To determine the remaining parameters, we numerically solve this
system of equations:
\begin{align*}
    \alpha m  &= \Fnominal \bonQone,\\
    m  &= \functionality(0)  -\Fnominal \bonQone e^{-\Qware_1 \tchange}+\Fnominal \bonQone.
\end{align*}
The first equation says that where the curve meets the minimum of the data, it has experienced exponential decay of $\alpha$ toward the asymptotic minimum.  We take $\alpha$ to be $\alpha=1-\sfrac{1}{e}$. The second equation says that the minimum occurs at the switching time (the time when the model switches from malware dominating bonware, to bonware dominating malware.  Solving this system of equations yields (with $\malware_1=\Qware_1-\bonware_1$), $\malware_1 \approx  0.025$ and $\bonware_1 \approx 0.005$, so that $\malwareEffectivenessConstant_1 = \sfrac{\malware_1}{\malwareActivityConstant_1} \approx 0.503$ and $\bonwareEffectivenessConstant_1 = \sfrac{\bonware_1}{\bonwareActivityConstant_1}\approx 0.362$.

To the right of $t^\star$, we fit an exponentially increasing function.  Similar to before the switching time, we compute the activities of the malware and bonware: $\malwareActivityConstant_2 \approx \frac{1}{100-69.5} \approx 0.033$ and $\bonwareActivityConstant_2 \approx \frac{4}{100-69.5} \approx 0.131$.

\newcommand\FBonQtwo[1]{\frac{#1 \bonware_2}{\Qware_2}}
To determine the remaining parameters, we numerically solve this system of equations:
  \begin{align*}
    \zeta &= \FBonQtwo{F(0)},\\
    \tilde{\alpha} \zeta &= \left(m- \FBonQtwo{\Fnominal}\right)
    \!\! \left(e^{-\Qware_2 (125-t^\star)}+\FBonQtwo{\Fnominal}\right).
  \end{align*}
  
We have found that $\tilde{\alpha}=1-e^{-4}$ and $\zeta=0.95$ are
satisfactory values to use for these hyperparameters.

We compute $\malware_2 \approx 0.005$ and $\bonware_2 \approx 0.088$,
so that
$\malwareEffectivenessConstant_1 =
\sfrac{\malware_1}{\malwareActivityConstant_1} \approx 0.201$ and
$\bonwareEffectivenessConstant_1 =
\sfrac{\bonware_1}{\bonwareActivityConstant_1}\approx 0.957$.

\subsection{Generating stochastic realizations}

Using the parameters found in Section~\ref{sec:parameters}, we can now
generate stochastic realizations.  In Section~\ref{sec:proposition},
we showed that for large sample sizes, the average of our ensemble
will approach the solution to the continuous ODE model.  We show
empirically that this is indeed the case.  To illustrate, we generated
five realizations (Figure \ref{fig:five:realizations}) of the
stochastic model with parameters found from the notional data of
Section~\ref{sec:parameters}.  
By averaging $n$ curves when
$n \in \{5,50,500,5000\}$ we see how the ensemble average approaches
the solution of the corresponding differential equation as predicted
in the Theorem of Section~\ref{sec:proposition}.


\begin{figure}[t]
\begin{center}
  \pppp{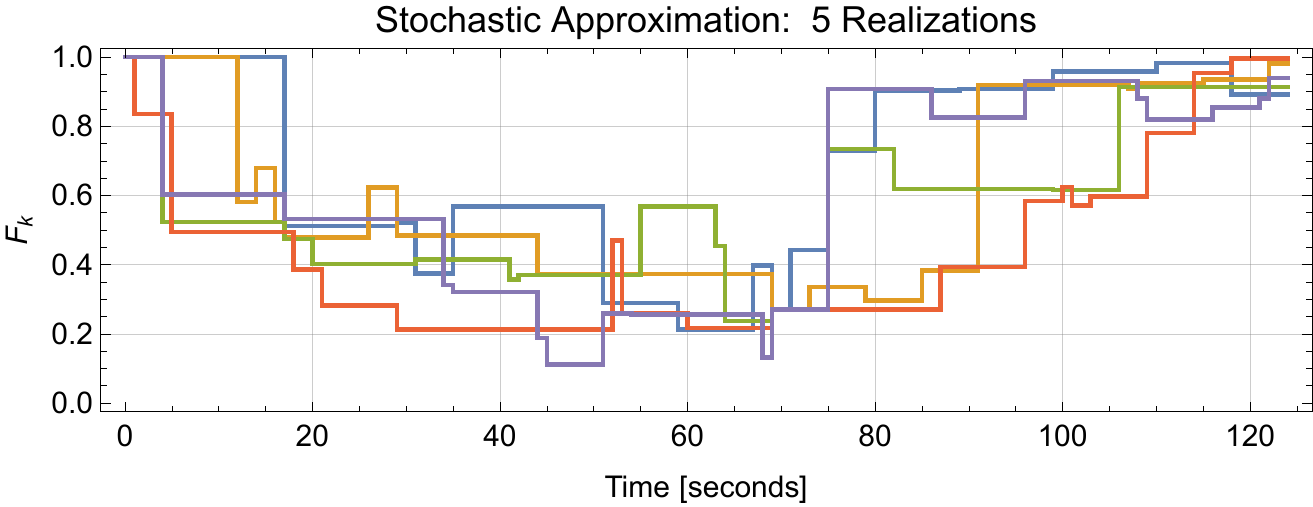}
  \ppp{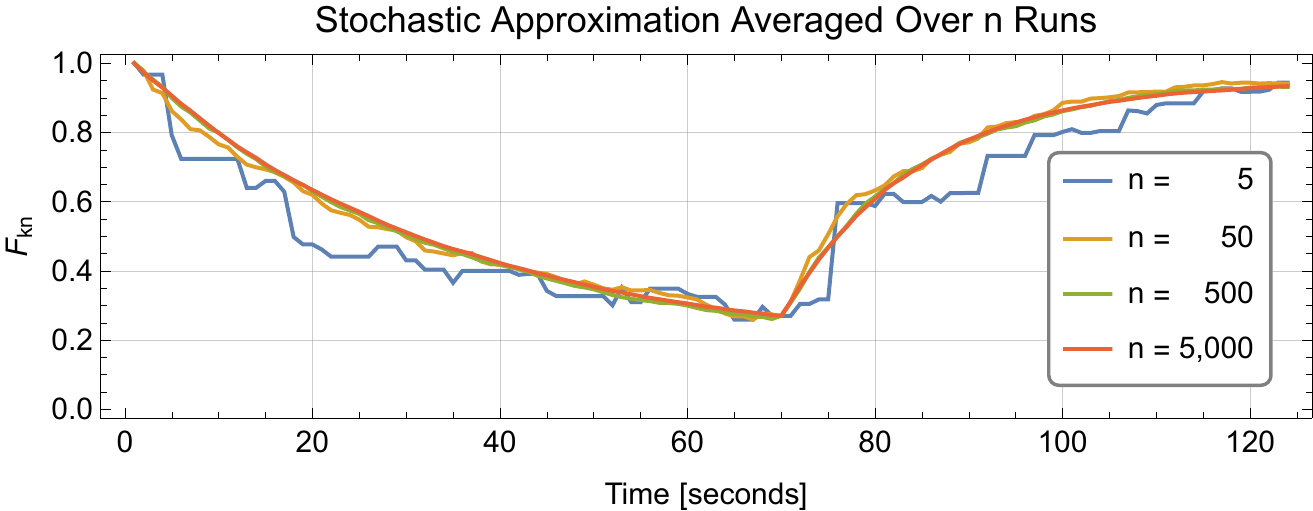}
 \caption{ (Top) Five realizations of the stochastic model generated with the
   parameters obtained by fitting the notional data shown in Figure
   \ref{fig:notional}.  Each realization is different but each roughly follows
   an exponential decay when $t<t^\star=69.5$ s and an
   exponential recovery for $t\ge t^\star$. (Bottom) Averages of $n$ stochastic runs for $n\in \{5,50,500,5000\}.$  As $n$ increases, the average of the ensemble approaches the fitted curve as predicted in the theorem of Section \ref{sec:proposition}.}
   \label{fig:five:realizations}
   \label{fig:average}
 \end{center}
 \end{figure}

\section{Discussion and conclusion}

We have presented a broadly applicable framework for the analysis of the cyber resilience of military artifacts.  Our framework relies on the construction of a custom differential equation time series model that shows good qualitative correspondence to the functionality of vehicles performing missions.  Seeking to move beyond the use of area-under-the-curve quantifications of cyber resilience, our proposed models have the advantage that their parameters have domain-relevant interpretations such as \textit{the activity of malware} and \textit{the effectiveness of bonware}.  Such interpretable parameters can provide a more nuanced interpretation of cyber resilience data 
being experimentally obtained in our lab.

Our formal models come in two families with complementary advantages. A series of continuous models is parsimonious, mathematically convenient, and easy to fit.  A series of discrete, stochastic models shows greater verisimilitude but is slightly less parsimonious and requires more computationally onerous parameter estimation techniques.

Both types of models can be extended to a large variety of custom circumstances, including the case where model parameters change gradually, abruptly, or predictably as a result of experimental manipulation.  Future work will include an extension to the cases of multiple simultaneous objectives and to the case of multiple vehicles to be analyzed jointly.

\bibliographystyle{unsrtnat}
\bibliography{references.bib}

\begin{thebibliography}{19}
\providecommand{\natexlab}[1]{#1}
\providecommand{\url}[1]{\texttt{#1}}
\expandafter\ifx\csname urlstyle\endcsname\relax
  \providecommand{\doi}[1]{doi: #1}\else
  \providecommand{\doi}{doi: \begingroup \urlstyle{rm}\Url}\fi

\bibitem[Kott and Linkov(2019)]{kott2019cyber}
A.~Kott and I.~Linkov.
\newblock \emph{Cyber resilience of systems and networks}.
\newblock Springer International Publishing, New York, NY, 2019.

\bibitem[Linkov et~al.(2018)Linkov, Trump, and Keisler]{linkov2018risk}
I.~Linkov, B.~D. Trump, and J.~Keisler.
\newblock Risk and resilience must be independently managed.
\newblock \emph{Nature}, 555:\penalty0 7694, 2018.

\bibitem[Bozdal et~al.(2018)Bozdal, Samie, and Jennions]{bozdal2018august}
M.~Bozdal, M.~Samie, and I.~Jennions.
\newblock A survey on {CAN} bus protocol: Attacks, challenges, and potential
  solutions.
\newblock In \emph{2018 International Conference on Computing, Electronics \&
  Communications Engineering}, pages 201--205. IEEE, August 2018.

\bibitem[Kott et~al.(2018)Kott, Th{\'e}ron, Dra{\v s}ar, Dushku, LeBlanc,
  Losiewicz, ..., and Rzadca]{kott2018}
A.~Kott, P.~Th{\'e}ron, M.~Dra{\v s}ar, E.~Dushku, B.~LeBlanc, P.~Losiewicz,
  ..., and K~Rzadca.
\newblock Autonomous intelligent cyber-defense agent (aica) reference
  architecture, 2018.
\newblock arXiv:1803.10664.

\bibitem[Kott et~al.(2021)Kott, Golan, Trump, and Linkov]{kott2021cyber}
A.~Kott, M.~S. Golan, B.~D. Trump, and I.~Linkov.
\newblock Cyber resilience: by design or by intervention?
\newblock \emph{Computer}, 54\penalty0 (8):\penalty0 112--117, 2021.

\bibitem[Kott and Th{\'e}ron(2020)]{kott2020doers}
A.~Kott and P.~Th{\'e}ron.
\newblock Doers, not watchers: Intelligent autonomous agents are a path to
  cyber resilience.
\newblock \emph{IEEE Security \& Privacy}, 18\penalty0 (3):\penalty0 62--66,
  2020.

\bibitem[Ellis et~al.(2022)Ellis, Parker, Vandekerckhove, Murphy, Smith, Kott,
  and Weisman]{ellis2022experimental}
J.~Ellis, T.~Parker, J.~Vandekerckhove, B.~Murphy, S.~Smith, A.~Kott, and
  M.~Weisman.
\newblock Experimental infrastructure for study of measurements of resilience.
\newblock \emph{Proceedings of IEEE Military Communications Conference}, pages
  841--846, Dec. 2022.

\bibitem[Alexeev et~al.(2017)Alexeev, Henshel, Levitt, McDaniel, B., Templeton,
  and Weisman]{alexeev}
A.~Alexeev, D.~Henshel, K.~Levitt, P.~McDaniel, Rivera B., S.~Templeton, and
  M.~Weisman.
\newblock Constructing a science of cyber-resilience for military systems.
\newblock \emph{NATO IST-153 Workshop on Cyber Resilience}, pages 23---25,
  2017.

\bibitem[Henshel et~al.(2019)Henshel, Levitt, Templeton, Cains, Alexeev,
  Blakely, McDaniel, Wehner, Rowell, and Weisman]{henshel}
D.~S. Henshel, K.~Levitt, S.~Templeton, M.~G. Cains, A.~Alexeev, B.~Blakely,
  P.~McDaniel, G.~Wehner, J.~Rowell, and M.~Weisman.
\newblock The science of cyber resilience: Characteristics and initial system
  taxonomy.
\newblock \emph{Fifth World Conference on Risk}, 2019.

\bibitem[Kott et~al.(2017)Kott, Ludwig, and Lange]{kott2017assessing}
A.~Kott, J.~Ludwig, and M.~Lange.
\newblock Assessing mission impact of cyberattacks: Toward a model-driven
  paradigm.
\newblock \emph{IEEE Secur Privacy}, 15\penalty0 (5):\penalty0 65--74, January
  2017.

\bibitem[Ligo et~al.(2021)Ligo, Kott, and Linkov]{ligo2021how}
A.~K. Ligo, A.~Kott, and I.~Linkov.
\newblock How to measure cyber-resilience of a system with autonomous agents:
  Approaches and challenges.
\newblock \emph{IEEE Engineering Management Review}, 49\penalty0 (2):\penalty0
  89--97, 2021.

\bibitem[Linkov et~al.(2013)Linkov, Eisenberg, Plourde, Seager, Allen, and
  Kott]{linkov2013resilience}
I.~Linkov, D.~A. Eisenberg, K.~Plourde, T.~P. Seager, J.~Allen, and A.~Kott.
\newblock Resilience metrics for cyber systems.
\newblock \emph{Environ. Syst. Decis.}, 33\penalty0 (4):\penalty0 471--476,
  2013.

\bibitem[Beling et~al.(2021)Beling, Horowitz, and
  McDermott]{beling2021developmental}
Peter Beling, Barry Horowitz, and Tom McDermott.
\newblock Developmental test and evaluation (dte\&a) and cyberattack resilient
  systems.
\newblock \emph{TR SERC-2021-TR-015 (V2)}, September 2021.

\bibitem[Hosseini et~al.(2016)Hosseini, Barker, and
  Ramirez-Marquez]{hosseini2016review}
S.~Hosseini, K.~Barker, and J.~E. Ramirez-Marquez.
\newblock A review of definitions and measures of system resilience.
\newblock \emph{Reliability Engineering \& System Safety}, 145:\penalty0
  47--61, 2016.

\bibitem[Kott and Linkov(2021)]{kott2021to}
A.~Kott and I.~Linkov.
\newblock To improve cyber resilience, measure it.
\newblock \emph{Computer}, 54\penalty0 (2):\penalty0 80--85, Feb. 2021.

\bibitem[Griffiths et~al.(2008)Griffiths, Kemp, and
  Tenenbaum]{griffiths2008bayesian}
T.~L. Griffiths, C.~Kemp, and J.~B. Tenenbaum.
\newblock \emph{{B}ayesian models of cognition}, pages 59--100.
\newblock Cambridge University Press, Cambridge, MA, 2008.

\bibitem[Plummer(2003)]{plummer2003jags}
M.~Plummer.
\newblock {JAGS}: {A} program for analysis of {B}ayesian graphical models using
  {G}ibbs sampling.
\newblock In K.~Hornik, F.~Leisch, and A.~Zeileis, editors, \emph{Proceedings
  of the 3rd International Workshop on Distributed Statistical Computing}. DSC
  2003, Vienna, Austria, 2003.

\bibitem[Liu et~al.(1998)Liu, Rubin, and Wu]{liu1998parameter}
C.~Liu, D.~B. Rubin, and Y.~N. Wu.
\newblock Parameter expansion to accelerate em: the px-em algorithm.
\newblock \emph{Biometrika}, 85\penalty0 (4):\penalty0 755--770, 1998.

\bibitem[Gelman(2004)]{gelman2004parameterization}
A.~Gelman.
\newblock Parameterization and {B}ayesian modeling.
\newblock \emph{Journal of the American Statistical Association}, 99:\penalty0
  537--545, 2004.

\end{thebibliography}

\end{document}